\documentclass[letterpaper, 10 pt, conference]{ieeeconf}  

\IEEEoverridecommandlockouts                              

\makeatletter
\def\endthebibliography{%
	\def\@noitemerr{\@latex@warning{Empty `thebibliography' environment}}%
	\endlist
}
\makeatother

\usepackage{mathrsfs} 

\usepackage{hyperref} 
\usepackage{graphicx,color}      
\usepackage{float}
\usepackage{cite}
\usepackage{booktabs}
\usepackage{amsmath, amssymb}
\usepackage{mathtools}

\DeclareMathAlphabet{\CMcal}{OMS}{cmsy}{m}{n} 

\newcommand\R{\ensuremath{\mathbb{R}}}

\newcommand\Mb{\ensuremath{\mathbf{M}}}

\newcommand\Ib{\ensuremath{\mathbf{I}}}

\newcommand\xb{\ensuremath{\mathbf{x}}}
\newcommand\wb{\ensuremath{\mathbf{w}}}
\newcommand\ub{\ensuremath{\mathbf{u}}}

\newcommand\yb{\ensuremath{\mathbf{y}}}

\newcommand\zb{\ensuremath{\mathbf{z}}}

\newcommand\argmin{\ensuremath{\operatorname*{argmin}}}

\let\epsilon\varepsilon

\newtheorem{theorem}{Theorem}
\newtheorem{lemma}{Lemma}
\newtheorem{definition}{Definition}
\newtheorem{proposition}{Proposition}

\newtheorem{corollary}{Corollary}

\title{\LARGE \bf
Perron-Frobenius Contractive Operator Matching for Data-Driven Reachable Fault Identification and Recovery
}

\author{Joshua D. Ibrahim, Mahdi Taheri, Soon-Jo Chung, Fred Y. Hadaegh
\thanks{Ph.D. Student, Control and Dynamical Systems, California Institute of Technology, Pasadena, CA. }%
\thanks{Postdoctoral Scholar, Computing and Mathematical Sciences, California Institute of Technology, Pasadena, CA.} 
\thanks{Bren Professor of Control and Dynamical Systems, California Institute of Technology, Pasadena, CA.}
\thanks{Research Professor in Aerospace, California Institute of Technology, Pasadena, CA.}
\thanks{Emails: \{\href{mailto:jdibrahi@caltech.edu}{jdibrahi}, \href{mailto:mtaheri@caltech.edu}{mtaheri},
 \href{mailto:sjchung@caltech.edu}{sjchung}, \href{mailto:hadaegh@caltech.edu}{hadaegh}\}@caltech.edu.}
\thanks{This research is funded in part by the DARPA Learning Introspective Control (LINC) program.}}
\renewcommand{\baselinestretch}{0.948}
\begin{document}

\maketitle
\thispagestyle{empty}
\pagestyle{empty}

\begin{abstract}
This paper focuses on data-driven fault detection, identification, and recovery (FDIR) for nonlinear control-affine systems under actuator faults. We create a unified framework in the space of probability densities, rather than on individual trajectories, using fault-indexed Perron--Frobenius (PF) operators to predict the evolution of state distributions under different fault profiles. By leveraging the probability-flow representation of the Fokker--Planck equation, we construct deterministic PF operators that reproduce exact stochastic marginals, define forward reachable density families, and establish certifiable 2-Wasserstein bounds on the divergence between fault-driven and nominal density evolutions. These provide quantitative conditions for the detectability and identifiability of various faults. The fault-indexed operators are learned from trajectory data via flow map matching (FMM), and we demonstrate that the observable FMM residual directly bounds the approximation error of the operator in the 2-Wasserstein metric. Additionally, we co-train a contraction certificate that bounds the gap between the learned operator family, the actual fault-driven density flow, and the nominal dynamics. The operator library is then used online for continuous fault parameter fitting over a continuous parameter space to generalize the learned operators to out-of-distribution (OOD) scenarios. To carry out the recovery control, we employ reachable density propagation and Gaussian mixture covariance steering. The proposed framework is validated on a 10-state spacecraft attitude-control system with four reaction wheels.

\end{abstract}
Autonomous systems operating in safety-critical environments, including spacecraft, uncrewed aerial vehicles (UAVs), and robotic platforms, must maintain reliable performance under unforeseen actuator and sensor faults~\cite{isermann2005fault,ragan2024online}. Faults that go undetected or unmitigated can lead to catastrophic failures, motivating integrated fault detection, identification, and recovery (FDIR) control frameworks suitable for real-time deployment~\cite{isermann2005fault,10502204}. While the problems of fault detection and identification (FDI) and fault-tolerant control (FTC) have each received significant attention independently, a unified framework that reasons about fault signatures and recovery control in a common mathematical language remains largely lacking in the literature.

Classical FDI approaches rely on residual generation using model-based estimators and parity equations and have been widely adopted in aerospace and industrial systems due to their interpretability and real-time feasibility~\cite{venkatasubramanian2003review,chen2012robust}. However, these methods assume accurate known system models and degrade in the presence of nonlinearities, simultaneous actuator and sensor faults, and stochastic uncertainties~\cite{talebi2008recurrent,fu2018evaluation}. On the fault recovery side, adaptive control methods offer strong stability guarantees when the fault structure is known a priori, but their effectiveness depends critically on accurate system modeling and knowledge of fault characteristics~\cite{isermann2005fault,10502204}. Data-driven controllers based on deep neural networks provide a promising alternative by learning complex system models directly from data~\cite{bakhtiaridoust2023data}. Yet, their performance can decline under distributional shift and out-of-distribution (OOD) scenarios at runtime~\cite{ibrahim2026probabilisticfaultdetection,farid2022task}. Density-based reachability~\cite{meng2022learning} learns reachable state distributions but does not address fault identification or provide such certificates.

A key limitation of many existing FDI and recovery approaches is that they operate primarily on individual trajectories rather than on the evolution of state distributions. At the distribution level, deterministic dynamics transport densities through the Liouville equation, whereas stochastic dynamics evolve densities according to the Fokker--Planck equation~\cite{brockett2007optimal,risken1989fokker}. Under this viewpoint, distinct fault profiles induce distinct density flows on the state space, and these flows may be compared using transport-based discrepancies such as the \(2\)-Wasserstein distance~\cite{benamou2000computational,villani2021topics}. Reasoning directly in density space naturally incorporates uncertainty, provides distribution-level fault signatures, and supports principled formulations of uncertainty-aware reachability, safety analysis, and recovery via Perron--Frobenius (PF) operator methods~\cite{jafarpour2024probabilistic,meng2022learning}.

In this work, we operate in the space of probability densities rather than on individual trajectories. The key idea is that each fault profile drives the state distribution along a distinct density flow, and these flows can be predicted, compared, and corrected using fault-indexed PF operators. We first construct deterministic PF operators that reproduce the stochastic density evolution by leveraging the probability-flow representation of the Fokker--Planck equation. We then learn these operators from offline trajectory data using flow map matching and co-train a contraction certificate to ensure stability over long prediction horizons. The learned operator library is deployed online to identify faults, including faults outside the training library, i.e., OOD scenarios, via continuous parameter fitting. Moreover, the learned operator library is utilized for recovery control, where the identified fault operator and the nominal operator are used to propagate reachable densities forward and steer the faulty distribution back toward the nominal target through Gaussian mixture covariance steering.

This density-level formulation addresses several limitations of existing approaches. Unlike classical residual-based FDI methods, which assume accurate known models and degrade under nonlinearities and stochastic uncertainties, the probability-flow construction converts the stochastic problem into a deterministic one that reproduces the stochastic marginals. Unlike trajectory-level data-driven controllers, which degrade under distributional shift, our framework provides certifiable $\mathbb W_2$ bounds on long-horizon stability via contraction regularization. We have three main contributions in this paper. First, we derive fault-indexed PF operators for stochastic fault dynamics via their probability-flow representation, establish explicit \(\mathbb W_2\) bounds relating fault-driven and nominal density evolutions, and learn the associated transport maps directly from trajectory data. Second, we prove that the observable endpoint flow-matching residual controls the deployed operator error in \(\mathbb W_2\), and that contraction regularization provides explicit long-horizon \(\mathbb W_2\) bounds between the learned operator family, the true fault-driven density flow, and the nominal system. Third, we show that the resulting operator library supports an online FDIR architecture for fault inference and recovery through reachable-density propagation and Gaussian-mixture-based tracking.
\section{Preliminaries}
We consider a safety-critical nonlinear system subject to unknown fault inputs trained from a finite library. Our goal is to detect and identify the active fault by tracking how each fault reshapes the state distribution over time, and to synthesize a corrective control that steers the distribution back toward its nominal trajectory. To this end, we work in the space of probability densities and characterize fault-induced distributional shifts via fault-indexed Perron-Frobenius operators, which propagate densities forward under each inferred fault profile.

Let \((\mathscr{X},\mathscr{B}(\mathscr{X}),\mu)\) be a measurable state space, where \(\mathscr{B}(\mathscr{X})\) is the Borel \(\sigma\)-algebra on \(\mathscr{X}\subseteq\mathbb{R}^n\) and \(\mu\) is a reference measure on \((\mathscr{X},\mathscr{B}(\mathscr{X}))\). Consider the nonlinear control-affine system
\begin{equation}
\label{eq:nonlinear_system}
\dot{\xb}=f(\xb,t)+g(\xb,t)\ub+\psi(\xb,t)\mathbf{w},
\end{equation}
where \(\xb:\R_{\ge0}\to\mathscr X\) is the state, \(\ub:\mathscr X\times\R_{\ge0}\to\R^m\) is the control input, and \(\mathbf{w}:\R_{\ge0}\to\R^p\) is the fault. Moreover, \(f:\mathscr X\times\R_{\ge0}\to\R^n\), \(g:\mathscr X\times\R_{\ge0}\to\R^{n\times m}\), and \(\psi:\mathscr X\times\R_{\ge0}\to\R^{n\times p}\) are sufficiently smooth, and \(\ub_{cl}:\mathscr X\times\R_{\ge0}\to\R^m\) is a nominal control feedback law. During the offline learning and inference stages, we have \(\ub=\ub_{cl}\). During recovery, an additive corrective input \(\ub_{rec}\) is synthesized and the applied control becomes \(\ub=\ub_{cl}+\ub_{rec}\).
The admissible fault library is \(\mathscr F=\{\mathbf w_0,\mathbf w_1,\ldots,\mathbf w_{N_f}\}\), with nominal, i.e., fault-free, profile \(\mathbf w_0\) defined as \(\mathbf w_0(t) = 0\). 

For each \(\mathbf w\in\mathscr F\), the closed-loop fault-indexed vector field \(\mathbf F_{\mathbf w}:\mathscr X\times\R_{\ge0}\to\R^n\) is expressed by
\(
\mathbf F_{\mathbf w}(\xb,t):=
f(\xb,t)+g(\xb,t)\ub_{cl}(\xb,t)+\psi(\xb,t)\mathbf{w}(t).
\)
Each \(\mathbf{F}_\wb\) defines a distinct closed-loop vector field with associated flow map. Under stochastic disturbance, the state density evolves according to the Fokker--Planck equation. Assume \(\mathbf F_{\mathbf w}\) generates a two-time flow map
\(
\Phi_{s,t}^{\mathbf F_{\mathbf w}}:\mathscr X\to\mathscr X,
\,
\Phi_{s,t}^{\mathbf F_{\mathbf w}}(\xb_s)=\xb_t \), for \(0\le s\le t
\).
Under stochastic disturbance, the system evolves as the It\^o SDE
\begin{equation}
\label{eq:nonlinear_sde}
d\xb=\mathbf F_{\mathbf w}(\xb,t)\,dt+\sigma(\xb,t)\,d\mathscr W(t),
\qquad
\xb_0\sim\rho_0,
\end{equation}
where \(\sigma:\mathscr X\times\R_{\ge0}\to\R^{n\times q}\) is the diffusion coefficient, \(\mathscr W:\R_{\ge0}\to\R^q\) is a standard \(q\)-dimensional Brownian motion, and \(\rho(\xb, 0) = \rho_0(\xb)\) is the initial density. We assume \(\exists\,L_1,L_2>0\) such that
\begin{equation*}
\begin{aligned}
&\|\mathbf F_{\mathbf w}(\xb,t)\|^2
+\|\sigma(\xb,t)\|_F^2\le L_1(1+\|\xb\|^2), \\
&\|\mathbf F_{\mathbf w}(\xb,t)
-\mathbf F_{\mathbf w}(\yb,t)\|
+\|\sigma(\xb,t)-\sigma(\yb,t)\|\le L_2\|\xb-\yb\|,
\end{aligned}
\end{equation*}
for all \(\xb,\yb\in\mathscr X\) and \(t\in\R_{\ge0}\), where $\|\cdot\|_F$ denotes the Frobenius norm. These conditions ensure the existence and uniqueness of strong solutions to~\eqref{eq:nonlinear_sde}~\cite{lasota2013chaos}. Let \(\rho(\xb,t)\) denote the density of~\eqref{eq:nonlinear_sde}. Rather than tracking individual trajectories, we track the evolution of the entire state density $\rho(\mathbf{x},t)$, which is formalized via the Perron-Frobenius operator.

\begin{definition}[Perron--Frobenius operators]
\label{def:pfo}
For deterministic dynamics \(\dot{\xb}=\mathbf F_{\mathbf w}(\xb,t)\) that generates flow map $\Phi_{s,t}^{\mathbf F_{\mathbf w}}$ for $0 \le s \le t$, the Perron--Frobenius operator~\cite{lasota2013chaos} on measures $\mathbf P_{s,t}^{\mathbf F_{\mathbf w}}:\CMcal P_2(\mathscr X)\to\CMcal P_2(\mathscr X)$, where \(\mathcal{P}_2(\mathscr{X})\) denotes the set of Borel probability measures on \(\mathscr{X}\) with finite second moments, is defined by pushforward along the flow map, $\mathbf P_{s,t}^{\mathbf F_{\mathbf w}}\mu := \bigl(\Phi_{s,t}^{\mathbf F_{\mathbf w}}\bigr)_{\!\#}\mu$. That is, for every Borel set \(A\subseteq\mathscr X\), 
\[
\bigl(\mathbf P_{s,t}^{\mathbf F_{\mathbf w}}\mu\bigr)(A)
=
\mu \Bigl(\bigl(\Phi_{s,t}^{\mathbf F_{\mathbf w}}\bigr)^{-1}(A)\Bigr),
\]
and bounded measurable \(\varphi:\mathscr X\to\R\), 
\[
\int \varphi\,d\bigl(\mathbf P_{s,t}^{\mathbf F_{\mathbf w}}\mu\bigr)
=
\int \varphi\circ\Phi_{s,t}^{\mathbf F_{\mathbf w}}\,d\mu.
\]
If \(\mu_s(d\xb)=\rho_s(\xb)\,d\xb\) with \(\rho_s\in\mathcal D(\mathscr X)\), the density-level Perron--Frobenius operator
\(\mathscr P_{s,t}^{\mathbf F_{\mathbf w}}:\mathcal D(\mathscr X)\to\mathcal D(\mathscr X)\)
is the Radon--Nikodym representation of \(\mathbf P_{s,t}^{\mathbf F_{\mathbf w}}\), $\bigl(\mathscr P_{s,t}^{\mathbf F_{\mathbf w}}\rho_s\bigr)(\xb)
:=
\frac{d \bigl(\mathbf P_{s,t}^{\mathbf F_{\mathbf w}}\mu_s\bigr)}{d\xb}(\xb)$, such that
\begin{equation*} 
\label{eq:pf_def}
\bigl(\mathscr P_{s,t}^{\mathbf F_{\mathbf w}}\rho_s\bigr)(\xb)
=
\rho_s\!\Bigl(
\bigl(\Phi_{s,t}^{\mathbf F_{\mathbf w}}\bigr)^{-1}(\xb)
\Bigr)
\Bigl|
\det D \bigl(\Phi_{s,t}^{\mathbf F_{\mathbf w}}\bigr)^{-1}(\xb)
\Bigr|,
\end{equation*}
whenever \(\Phi_{s,t}^{\mathbf F_{\mathbf w}}\) is a \(C^1\)-diffeomorphism.
\end{definition}

To measure the separation between densities propagated under different fault hypotheses, which is the foundation of fault detectability, we equip $\mathcal{P}_2(\mathscr{X})$ with the 2-Wasserstein metric.

\begin{definition}[2-Wasserstein metric]
\label{def:w2}
Let \(\CMcal P_2(\mathscr X)\) be the set of Borel probability measures on \(\mathscr X\) with finite second moments, and let \(\Pi(\mu,\nu)\) denote the set of measures on $\mathscr{X} \times \mathscr{X}$ with marginals \(\mu,\nu\in\CMcal P_2(\mathscr X)\). The \(2\)-Wasserstein distance is
\begin{equation*}
\label{eq:w2_definition}
\mathbb W_2(\mu,\nu)
:=
\left(
\inf_{\pi\in\Pi(\mu,\nu)}
\int_{\mathscr X\times\mathscr X}
\|\xb-\yb\|^2\,d\pi(\xb,\yb)
\right)^{1/2}.
\end{equation*}
If \(\mu,\nu\in\CMcal P_2(\mathscr X)\) are absolutely continuous with respect to Lebesgue measure, then the Radon--Nikodym theorem implies there exists densities \(\rho=\frac{d\mu}{d\xb}\) and \(\eta=\frac{d\nu}{d\xb}\) such that \(\mu(d\xb)=\rho(\xb)\,d\xb\) and \(\nu(d\xb)=\eta(\xb)\,d\xb\). In that case we also write \(\mathbb W_2(\rho,\eta)\). We also denote the space of admissible probability densities on $\mathscr{X}$ with finite second moments by
\[
\mathcal D(\mathscr X)
:=
\left\{
\rho\in L^1(\mathscr X):
\begin{array}{l}
\rho\ge0,\;
\int_{\mathscr X}\rho(\xb)\,d\xb=1,\\[1mm]
\int_{\mathscr X}\|\xb\|^2\rho(\xb)\,d\xb<\infty
\end{array}
\right\},
\]
where \(L^1(\mathscr{X})\) denotes the space of Lebesgue-integrable functions on \(\mathscr{X}\).
\end{definition}

Both families \(\{\mathbf P_{s,t}^{\mathbf F_{\mathbf w}}\}_{0\le s\le t}\) and \(\{\mathscr P_{s,t}^{\mathbf F_{\mathbf w}}\}_{0\le s\le t}\) satisfy the two-time semigroup property
\[
\mathscr P_{s,s}^{\mathbf F_{\mathbf w}}=I,
\qquad
\mathscr P_{r,t}^{\mathbf F_{\mathbf w}}
=
\mathscr P_{s,t}^{\mathbf F_{\mathbf w}}
\mathscr P_{r,s}^{\mathbf F_{\mathbf w}},
\qquad 0\le r\le s\le t.
\]
with the same identity for \(\mathbf P_{s,t}^{\mathbf F_{\mathbf w}}\). 

In practice, the system is subject to stochastic disturbance. Hence, $\mathbf{x}$ is governed by the It\^{o} SDE \eqref{eq:nonlinear_sde} and the state density satisfies the Fokker--Planck equation below. Crucially, this stochastic density evolution can be exactly reproduced by the deterministic probability-flow ODE, whose pushforward defines 
a fault-indexed PF operator (PFO). The density \(\rho\) satisfies the Fokker--Planck equation
\begin{equation}
\label{eq:fokker_planck}
\frac{\partial \rho}{\partial t}
=
-\nabla_{\xb}\cdot\!\bigl(\rho\,\mathbf F_{\mathbf w}\bigr)
+\frac{1}{2}\nabla_{\xb}\cdot\nabla_{\xb}\cdot(\rho\,\Sigma),
\qquad
\Sigma:=\sigma\sigma^\top.
\end{equation}
\begin{theorem}[Probability Flow PF Operator]
\label{thm:probability_flow_pfo}
Let \(\mathbf w\in\mathscr F\). Consider \(\rho(\xb,t)\) as a strictly positive solution of~\eqref{eq:fokker_planck} on \([s,T]\) with \(\rho(\xb,s)=\rho_s(\xb)\). Define
\begin{equation}
\label{eq:probability_flow_ode}
\begin{aligned}
\mathbf v_{\mathbf w}(\xb,t)
:={}&
\mathbf F_{\mathbf w} -\frac{1}{2}\Bigl[(\nabla_\xb\cdot\Sigma)
+\Sigma\nabla_\xb\log\rho\Bigr],
\end{aligned}
\end{equation}
where $\left[\left(\nabla_{\mathbf{x}} \cdot \Sigma\right)(\mathbf{x}, t)\right]_i=\sum_{j=1}^n \frac{\partial \Sigma_{i j}}{\partial x_j}(\mathbf{x}, t)$, for $i=1, \ldots, n$ denotes the row-wise divergence of $\Sigma$. Assume \(\mathbf v_{\mathbf w}\) generates the deterministic flow $\Phi_{s,t}^{\mathbf v_{\mathbf w}}:\mathscr X\to\mathscr X$ with PFO $\mathscr P_{s,t}^{\mathbf v_{\mathbf w}}:\mathcal D(\mathscr X)\to\mathcal D(\mathscr X)$. Consequently, \(\rho\) satisfies
\begin{equation}
\label{eq:continuity_equation_probability_flow}
\frac{\partial \rho}{\partial t}
+\nabla_\xb\cdot\!\bigl(\mathbf v_{\mathbf w}(\xb,t)\rho(\xb,t)\bigr)=0,
\end{equation}
and therefore, \(\rho(\cdot,t)=\mathscr P_{s,t}^{\mathbf v_{\mathbf w}}\rho_s\), for \(s\le t\le T\). Hence, the deterministic probability flow PFO reproduces the marginals of the stochastic fault dynamics.
\end{theorem}

\begin{proof}
Rewrite~\eqref{eq:fokker_planck} as
\[
\frac{\partial \rho}{\partial t}
=
-\nabla_\xb\cdot(\mathbf F_{\mathbf w}\rho)
+\frac{1}{2}\nabla_\xb\cdot\!\bigl(\nabla_\xb\cdot(\Sigma\rho)\bigr).
\]
Considering the product rule, we have \(
\nabla_\xb\cdot(\Sigma\rho)
=
(\nabla_\xb\cdot\Sigma)\rho+\Sigma\nabla_\xb\rho.
\)
Since \(\rho>0\) and \(\nabla_\xb\rho=\rho\nabla_\xb\log\rho\), one obtains
\[
\nabla_\xb\cdot(\Sigma\rho)
=
\Bigl[(\nabla_\xb\cdot\Sigma)+\Sigma\nabla_\xb\log\rho\Bigr]\rho.
\]
Substituting into the Fokker--Planck equation yields
\[
\frac{\partial \rho}{\partial t}
=
-\nabla_\xb\cdot\!\Bigl[
\Bigl(
\mathbf F_{\mathbf w}
-\frac{1}{2}\bigl[(\nabla_\xb\cdot\Sigma)+\Sigma\nabla_\xb\log\rho\bigr]
\Bigr)\rho
\Bigr],
\]
which is exactly~\eqref{eq:continuity_equation_probability_flow} by \eqref{eq:probability_flow_ode}. Thus, \(\rho\) satisfies the Liouville equation driven by \(\mathbf v_{\mathbf w}\) with the initial condition
\(\rho_s\)~\cite{boffi2023probability}. By uniqueness of solutions to that transport equation~\cite{ambrosio2014continuity}, one has \(\rho(\cdot,t)=\mathscr P_{s,t}^{\mathbf v_{\mathbf w}}\rho_s\) for all \(t\in[s,T]\).
\end{proof}

For the probability flow $\dot{\xb}=\mathbf{v}_{\mathbf w}(\xb,t)$ of~\eqref{eq:nonlinear_sde} (Theorem~\ref{thm:probability_flow_pfo}), combining~\eqref{eq:probability_flow_ode} with the decomposition of $\mathbf{F}_{\mathbf{w}}$ yields the perturbed form
\begin{align}
\label{eq:fault_decomposition_density}
\mathbf{F}_{\mathbf{w}}(\xb,t) &= \mathbf{F}_{\mathbf{w}_0}(\xb,t)+\psi(\xb,t)\mathbf{w}(t),\\
\label{eq:probability_flow_as_deterministic_perturbation}
\dot{\xb} &= \mathbf{F}_{\mathbf{w}_0}(\xb,t)+\mathbf{d}^{\mathbf{w}}(\xb,t),\\
\label{eq:probability_flow_perturbation_term}
\mathbf{d}^{\mathbf{w}}(\xb,t) &:= \psi\mathbf{w}-\tfrac{1}{2}(\nabla_{\xb}\!\cdot\Sigma)-\tfrac{1}{2}\Sigma\nabla_{\xb}\log\rho.
\end{align}
Let $\xb^{\mathbf F_{\mathbf{w}_0}}$ and $\xb^{\mathbf v_{\mathbf w}}$ solve~\eqref{eq:nominal_closed_loop_dynamics} and~\eqref{eq:probability_flow_as_deterministic_perturbation}, respectively, and let $\mathbf{q}(\lambda,t)$, $\lambda\in[0,1]$, smoothly interpolate $\mathbf{q}(0,t)=\xb^{\mathbf F_{\mathbf{w}_0}}(t)$ to $\mathbf{q}(1,t)=\xb^{\mathbf v_{\mathbf w}}(t)$. A family of particular solutions is
\begin{equation}
\label{eq:particular_solutions_virtual_system}
\dot{\mathbf{q}}(\lambda,t)=\mathbf{F}_{\mathbf{w}_0}(\mathbf{q}(\lambda,t),t)+\lambda\,\mathbf{d}^{\mathbf{w}}(\xb^{\mathbf v_{\mathbf w}}(t),t).
\end{equation}
Thus, the probability-flow ODE is a deterministic perturbation of the nominal closed-loop dynamics, which motivates the use of contraction theory to analyze the gap between these two systems and to regularize operator learning during training.
\subsection{Wasserstein Contraction of Density}
Consider the nominal closed-loop dynamics of~\eqref{eq:nonlinear_system}, i.e., $\mathbf{w} = 0$, given by
\begin{equation}
\label{eq:nominal_closed_loop_dynamics}
\dot{\xb}=\mathbf{F}_{\mathbf{w}_0}(\xb,t), \ \delta \dot{\xb}=\frac{\partial \mathbf{F}_{\mathbf{w}_0}}{\partial \xb}\delta \xb, \ \xb(0)=\xb_0,
\end{equation}
where $\delta \xb$ is the infinitesimal displacement of any two trajectories of~\eqref{eq:nominal_closed_loop_dynamics} defined as $\xb_2 - \xb_1 = \int_{\xb_1}^{\xb_2}\delta \xb$, and $\mathbf{F}_{\mathbf{w}_0}(\xb,t):= f(\xb,t)+ g(\xb,t) \mathbf{u}_{c l}$ with some desired closed-loop controller $\ub_{c l}$. 

\begin{lemma}[Deterministic Contraction]
\label{lem:deterministic_contraction}
If there exists a uniformly positive definite metric $\mathbf{M}(\xb,t) = \Theta(\xb,t)^\top \Theta(\xb,t) \succ 0$, $\forall \xb, t$, with a smooth coordinate transformation of the virtual displacement $\delta \zb=\Theta(\xb,t) \delta \xb$, and for $\alpha \in \R$, such that
\begin{equation}
\label{eq:deterministic_contraction_condition}
\dot{\mathbf{M}}(\xb,t)+2\,\operatorname{sym}\Bigl(\mathbf{M}(\xb,t)
\frac{\partial \mathbf{F}_{\mathbf{w}_0}}{\partial \xb}\Bigr)
\preceq 2\alpha \mathbf{M}(\xb,t),
\end{equation}
$\forall \xb,t$, where $\operatorname{sym}(\mathbf{A}) := \frac{1}{2}\left(\mathbf{A}+\mathbf{A}^\top\right)$ for any square matrix \(\mathbf{A}\), holds for \eqref{eq:nominal_closed_loop_dynamics}, then any two trajectories of~\eqref{eq:nominal_closed_loop_dynamics} are bounded according to \(\|\delta \zb(t)\| = \|\Theta(\xb,t) \delta \xb(t)\| \leq\|\delta \zb(0)\| e^{\alpha t}\). If $\alpha < 0$, the system is said to be contracting.

\begin{proof}
See~\cite{Lohmiller1998,tsukamoto2021contraction} for more details.
\end{proof}
\end{lemma}
\begin{theorem}[Contraction of Probability Flow]
  \label{thm:wasserstein_contraction_bound} 
Assume that \eqref{eq:nominal_closed_loop_dynamics} satisfies the deterministic contraction condition~\eqref{eq:deterministic_contraction_condition} of Lemma~\ref{lem:deterministic_contraction}, where $\exists \underline{m}, \overline{m}\in \R_{> 0}$, such that $\underline{m}\mathbf{I} \preceq \mathbf{M}(\xb, t) \preceq \overline{m}\mathbf{I}, \forall \xb, t$. Suppose for the perturbed dynamics~\eqref{eq:probability_flow_as_deterministic_perturbation}, $\exists \overline{d} \in \R_{> 0}$ such that $\overline{d} = \sup_{\xb, t} \|\mathbf{d}^{\mathbf{w}}(\xb,t)\| < \infty$, then the Wasserstein distance between the distributions $\mu_t$ in~\eqref{eq:nominal_closed_loop_dynamics} and $\nu_t$ of~\eqref{eq:probability_flow_as_deterministic_perturbation} with initial conditions $\xb^{\mathbf F_{\mathbf{w}_0}}(0) \sim \mu_0$ and $\xb^{\mathbf v_{\mathbf w}}(0) \sim \nu_0$ is bounded as
\begin{equation}
\label{eq:wasserstein_contraction_bound}
\mathbb{W}_2\left(\mu_t, \nu_t\right) \leq \kappa e^{\alpha t} \mathbb{W}_2\left(\mu_0, \nu_0\right)+\delta_t,
\end{equation}
where $\kappa:=\sqrt{\overline{m} / \underline{m}}$ and $\delta_t:=\kappa \overline{d} \int_0^t e^{\alpha(t-r)} d r$.

\begin{proof}
Using the contraction condition~\eqref{eq:deterministic_contraction_condition} of Lemma~\ref{lem:deterministic_contraction}, for $\mathbf{M}(\xb, t) = \Theta(\xb, t)^\top \Theta(\xb, t)$, we have
\begin{equation*}
  \begin{split}
  \frac{d}{d t}\| \Theta(\mathbf{q}, t) \partial_\lambda \mathbf{q} \| &= (2 \|\Theta(\mathbf{q}, t) \partial_\lambda \mathbf{q} \|)^{-1} \frac{d}{d t} \partial_\lambda \mathbf{q}^{\top} \mathbf{M}(\mathbf{q}, t) \partial_\lambda \mathbf{q} \\
  &\leq \alpha \|\Theta(\mathbf{q}, t) \partial_\lambda \mathbf{q} \| + \|\Theta(\mathbf{q}, t) \partial_\lambda \mathbf{d}_\lambda\|,
  \end{split}
\end{equation*}
where $\partial_\lambda \mathbf{q}=\partial \mathbf{q} / \partial \lambda$ and $\partial_\lambda \mathbf{d}_\lambda=\partial \mathbf{d}_\lambda / \partial \lambda=\mathbf{d}^{\mathbf{w}}\left(\xb^{\mathbf v_{\mathbf w}}, t\right)$. Let $V(\mathbf{q}, \delta\mathbf{q}, t) = \int_{\xb^{\mathbf F_{\mathbf{w}_0}}}^{\xb^{\mathbf v_{\mathbf w}}} \|\Theta(\mathbf{q}, t) \delta\mathbf{q}\|$ be the generalized length with respect to the Riemannian metric $\mathbf{M}(\mathbf{q}, t)$. Taking the integral with respect to $\lambda$ gives
\begin{equation*}
\frac{d}{d t} \int_0^1\left\|\Theta \partial_\lambda \mathbf{q}\right\| d \lambda \leq \int_0^1 \alpha\left\|\Theta \partial_\lambda \mathbf{q}\right\|+\left\|\Theta \mathbf{d}^{\mathbf{w}}\left(\xb^{\mathbf v_{\mathbf w}}, t\right)\right\| d \lambda ,
\end{equation*}
which implies $\dot{V}(\mathbf{q}, \delta\mathbf{q}, t) \leq \alpha V(\mathbf{q}, \delta\mathbf{q}, t) + \sup_{\mathbf{q},\xb^{\mathbf v_{\mathbf w}}, t}\|\Theta(\mathbf{q}, t) \mathbf{d}^{\mathbf{w}}\left(\xb^{\mathbf v_{\mathbf w}}, t\right)\|$. By Gronwall's inequality,
\begin{equation*}
  V(\mathbf{q}, \delta\mathbf{q}, t) \leq e^{\alpha t} V(\mathbf{q}, \delta\mathbf{q}, 0) + \sqrt{\overline{m}} \overline{d} \int_0^t e^{\alpha(t-r)} d r.
\end{equation*}
Now define the Riemannian distance induced by $\mathbf{M}$ as
\begin{equation*}
  d_\Mb(\xb^{\mathbf F_{\mathbf{w}_0}}(t), \xb^{\mathbf v_{\mathbf w}}(t), t) := \inf_{\mathbf{q} \in \Gamma(\xb^{\mathbf F_{\mathbf{w}_0}}(t), \xb^{\mathbf v_{\mathbf w}}(t))} V(\mathbf{q}, \delta\mathbf{q}, t),
\end{equation*}
where $\Gamma(\xb^{\mathbf F_{\mathbf{w}_0}}(t), \xb^{\mathbf v_{\mathbf w}}(t))$ denotes the set of smooth paths connecting $\xb^{\mathbf F_{\mathbf{w}_0}}(t)$ and $\xb^{\mathbf v_{\mathbf w}}(t)$. Since the above bound holds for every admissible path, it also holds for the infimum. Using the metric bounds $\sqrt{\underline{m}}\left\|\xb^{\mathbf F_{\mathbf{w}_0}}(t)-\xb^{\mathbf v_{\mathbf w}}(t)\right\| \leq d_\Mb(\xb^{\mathbf F_{\mathbf{w}_0}}(t), \xb^{\mathbf v_{\mathbf w}}(t), t) \leq \sqrt{\overline{m}}\left\|\xb^{\mathbf F_{\mathbf{w}_0}}(t)-\xb^{\mathbf v_{\mathbf w}}(t)\right\|$, $\forall t$, we have
\begin{equation}
\label{eq:wasserstein_contraction_bound_pointwise}
\| \xb^{\mathbf F_{\mathbf{w}_0}}(t)-\xb^{\mathbf v_{\mathbf w}}(t)\| \leq \kappa e^{\alpha t} \| \xb^{\mathbf F_{\mathbf{w}_0}}(0)-\xb^{\mathbf v_{\mathbf w}}(0)\|+\delta_t.
\end{equation}
Take $\pi_0\in\Pi(\mu_0,\nu_0)$, $(\mathbf{X}_0,\mathbf{Y}_0)\sim\pi_0$, and $\mathbf{X}_t := \Phi_{0,t}^{\mathbf F_{\mathbf{w}_0}}(\mathbf{X}_0)$, $\mathbf{Y}_t := \Phi_{0,t}^{\mathbf v_{\mathbf w}}(\mathbf{Y}_0)$, so $\mu_t = (\Phi_{0,t}^{\mathbf F_{\mathbf{w}_0}})_\# \mu_0$ and $\nu_t = (\Phi_{0,t}^{\mathbf v_{\mathbf w}})_\# \nu_0$. Then $\pi_t:=\operatorname{Law}(\mathbf{X}_t,\mathbf{Y}_t)\in\Pi(\mu_t,\nu_t)$. Since the pointwise bound~\eqref{eq:wasserstein_contraction_bound_pointwise} holds $\pi_0$-a.s., taking $L^2$ norms and applying Minkowski, then using $\mathbb{W}_2(\mu_t,\nu_t) \leq \bigl(\mathbb{E}\|\mathbf{X}_t-\mathbf{Y}_t\|^2\bigr)^{1/2}$ and the infimum over $\pi_0$, yields
\begin{align*}
\|\mathbf{X}_t-\mathbf{Y}_t\| &\leq \kappa\,e^{\alpha t}\|\mathbf{X}_0-\mathbf{Y}_0\|+\delta_t,\\
\bigl(\mathbb{E}\|\mathbf{X}_t-\mathbf{Y}_t\|^2\bigr)^{1/2} &\leq \kappa\,e^{\alpha t}\bigl(\mathbb{E}\|\mathbf{X}_0-\mathbf{Y}_0\|^2\bigr)^{1/2}+\delta_t,\\
\mathbb{W}_2(\mu_t,\nu_t) &\leq \kappa\,e^{\alpha t}\mathbb{W}_2(\mu_0,\nu_0)+\delta_t.
\end{align*}
\end{proof} 
\end{theorem}

\subsection{Density Reachability and Fault Detectability}
Theorem~\ref{thm:wasserstein_contraction_bound} quantifies how fast the distributional gap between the nominal and faulty densities grows (or decays) as a function of the perturbation magnitude $\overline{d}$ and the contraction rate $\alpha$. We now use this bound to formalize when the fault-induced density shift is large enough to be reliably detected and when two distinct faults produce sufficiently separated density evolutions to be distinguished from one another.

\begin{definition}[Forward Reachable Densities]
\label{def:forward_reachable}
Let \(\mathfrak D_s\subseteq\mathcal D(\mathscr X)\) be a family of admissible initial densities. The forward reachable density family under fault profile \(\mathbf w\) is
\[
\mathfrak R_{s,t}^{+}(\mathfrak D_s,\mathbf w)
:=
\left\{
\mathscr P_{s,t}^{\mathbf v_{\mathbf w}}\rho_s:\rho_s\in\mathfrak D_s
\right\}.
\]
When the initial density is fixed as \(\rho_s\in\mathcal D(\mathscr X)\), the fault-indexed reachable density family is
\begin{equation}
\label{eq:reachable_density_family}
\mathfrak R_{s,t}^{+}(\rho_s)
:=
\left\{
\mathscr P_{s,t}^{\mathbf v_{\mathbf w}}\rho_s:\mathbf w\in\mathscr F
\right\}.
\end{equation}
\end{definition}

The PFO on measures \(\mathbf P_{s,t}^{\mathbf F_{\mathbf w}}\) for the above definitions can be shown to recover classical state reachability by considering and initial law as Dirac masses. The classical reachable set is recovered exactly as a special case, while the density framework additionally quantifies uncertainty over initial conditions and stochastic disturbances, generalizing set-based reachability.

\begin{proposition}[Dirac recovery of state reachability]
\label{prop:dirac_recovery}
Fix \(\mathbf w\in\mathscr F\), let
\(\mathscr R_{s,t}^{+}\left(\mathcal S_s\right)
:=
\Phi_{s,t}^{\mathbf F_{\mathbf w}}\left(\mathcal S_s\right)\), 
and define \(\mathfrak E(\mathcal S):=\{\delta_{\xb}:\xb\in\mathcal S\}\), where \(\delta_\xb\) denotes the Dirac measure at \(\xb\). Consequently,
\(
\mathfrak E \Bigl(\Phi_{s,t}^{\mathbf F_{\mathbf w}}(\mathcal S_s)\Bigr)
=
\Bigl\{
\mathbf P_{s,t}^{\mathbf F_{\mathbf w}}\delta_{\xb}:\xb\in\mathcal S_s
\Bigr\},
\)
such that the full reachable set is recovered exactly from the Dirac embedding of
\(\mathcal S_s\).
\end{proposition}
  
\begin{proof}
Let \(\varphi:\mathscr X\to\R\) be any bounded measurable test function. For any \(\xb\in\mathcal S_s\), one has
\begin{equation*}
\begin{split}
\int_{\mathscr X}
\varphi(\yb)\,
d\!\left(
\mathbf P_{s,t}^{\mathbf F_{\mathbf w}}\delta_{\xb}
\right)(\yb)
&=
\int_{\mathscr X}
\varphi\!\left(
\Phi_{s,t}^{\mathbf F_{\mathbf w}}(\yb)
\right)
d\delta_{\xb}(\yb) \\
&=
\varphi\!\left(
\Phi_{s,t}^{\mathbf F_{\mathbf w}}(\xb)
\right).
\end{split}
\end{equation*}
On the other hand, we have
\begin{equation*}
\int_{\mathscr X}\varphi(\yb)\, d\delta_{\Phi_{s,t}^{\mathbf F_{\mathbf w}}(\xb)}(\yb)
=
\varphi\!\left( \Phi_{s,t}^{\mathbf F_{\mathbf w}}(\xb) \right).
\end{equation*}
Thus, 
\(
\mathbf P_{s,t}^{\mathbf F_{\mathbf w}}\delta_{\xb}
=
\delta_{\Phi_{s,t}^{\mathbf F_{\mathbf w}}(\xb)}
\)
for every \(\xb\in\mathcal S_s\). Hence, we obtain $\{
\delta_{\Phi_{s,t}^{\mathbf F_{\mathbf w}}(\xb)}:\xb\in\mathcal S_s
\} = \mathfrak E\!\left(\Phi_{s,t}^{\mathbf F_{\mathbf w}}(\mathcal S_s)
\right)$. Since \(\Phi_{s,t}^{\mathbf F_{\mathbf w}}(\mathcal S_s)=\mathscr R_{s,t}^{+}(\mathcal S_s)\) by definition, the identity follows immediately.
\end{proof}
\begin{definition}[Density detectability and identifiability]
\label{def:density_reachability}
Fix \(\rho_s\in\mathcal D(\mathscr X)\), an interval \([s,t]\), and \(\epsilon>0\). A fault \(\mathbf w\neq\mathbf w_0\) is \(\epsilon\)-detectable in density on \([s,t]\) from \(\rho_s\) if
\begin{equation}
\label{eq:density_detectability}
\mathbb W_2\!\left(
\mathscr P_{s,t}^{\mathbf v_{\mathbf w}}\rho_s,\;
\mathscr P_{s,t}^{\mathbf v_{\mathbf w_0}}\rho_s
\right)\ge\epsilon.
\end{equation}
Two fault profiles \(\mathbf w_i,\mathbf w_j\in\mathscr F\) are \(\epsilon\)-identifiable in density on \([s,t]\) from \(\rho_s\) if
\begin{equation}
\label{eq:density_identifiability}
\mathbb W_2\!\left(
\mathscr P_{s,t}^{\mathbf v_{\mathbf w_i}}\rho_s,\;
\mathscr P_{s,t}^{\mathbf v_{\mathbf w_j}}\rho_s
\right)\ge\epsilon.
\end{equation}
\end{definition}
\begin{theorem}[Fault detectability bound]
\label{thm:detect_bound}
Fix \(\mathbf w_i,\mathbf w_j\in\mathscr F\). Let
\(s^{\mathbf w}(\xb,t):=\nabla_\xb\log\rho^{\mathbf w}(\xb,t)\) denote the
score of the fault-indexed marginal along~\eqref{eq:nonlinear_sde}, and assume
\(\exists\,\bar\psi,\bar\sigma,\bar s_{ij}<\infty\), such that
\begin{equation}
  \label{eq:score_bounds}
  \begin{split}
\sup_{\xb,t}\|\psi(\xb,t)\|_F \le \bar\psi, \quad \sup_{\xb,t}\|\Sigma(\xb,t)\|_F \le \bar\sigma,\\
\sup_{\xb,t}\bigl\|s^{\mathbf w_{j}}(\xb,t)-s^{\mathbf w_{i}}(\xb,t)\bigr\|
\le \bar s_{ij}.
\end{split}
\end{equation}
Let \(\rho^{\mathbf w_{i}}(\cdot,t)\), \(\rho^{\mathbf w_{j}}(\cdot,t)\) be the marginals of~\eqref{eq:nonlinear_sde} with
\(\rho^{\mathbf w_{i}}(\cdot,0)=\rho^{\mathbf w_{j}}(\cdot,0)=\rho_0\).
Suppose the probability-flow dynamics for \(\mathbf w_i\) are contracting in
\(\mathbf M^{ij}(\xb,t)=(\Theta^{ij})^\top\Theta^{ij}\succ0\) with rate
\(\alpha_{ij}\) and
\(\underline m^{ij}\Ib\preceq \mathbf M^{ij}(\xb,t)\preceq\bar m^{ij}\Ib\).
Set \(\Delta w_{ij}:=\sup_{t\ge0}\|\mathbf w_i(t)-\mathbf w_j(t)\|\) and
\(\kappa^{ij}:=\sqrt{\bar m^{ij}/\underline m^{ij}}\), and define
\begin{equation*}
\label{eq:d_bar_ij}
\bar d^{ij}
:=
\bar\psi\,\Delta w_{ij}
+\tfrac{1}{2}\bar\sigma\,\bar s_{ij}.
\end{equation*}
For all \(t\ge0\), one has
\begin{equation}
\label{eq:detect_upper}
\mathbb W_2\!\left(
\rho^{\mathbf w_{i}}(\cdot,t),\rho^{\mathbf w_{j}}(\cdot,t)
\right)
\le
\kappa^{ij}\bar d^{ij}
\int_0^t e^{\alpha_{ij}(t-r)}\,dr.
\end{equation}
Identifiability with margin \(\epsilon\) in the sense of
Definition~\ref{def:density_reachability} and~\eqref{eq:density_identifiability}
on \([0,T]\) is possible only if
\begin{equation}
\label{eq:detect_cert}
\epsilon
\le
\kappa^{ij}\bar d^{ij}
\int_0^T e^{\alpha_{ij}(T-r)}\,dr.
\end{equation}
If \(\alpha_{ij}<0\), then
\(\mathbb W_2(\rho^{\mathbf w_{i}}_\infty,\rho^{\mathbf w_{j}}_\infty)
\le \kappa^{ij}\bar d^{ij}/|\alpha_{ij}|\).
\end{theorem}

\begin{proof}
By~\eqref{eq:probability_flow_ode} and~\eqref{eq:fault_decomposition_density}, one obtains
\(\mathbf v_{\mathbf w_j}=\mathbf v_{\mathbf w_i}+\mathbf d^{ij}\) with
\(\mathbf d^{ij}
=\psi(\xb,t)[\mathbf w_j(t)-\mathbf w_i(t)]
+\tfrac{1}{2}\Sigma(\xb,t)\bigl[
s^{\mathbf w_{i}}(\xb,t)-s^{\mathbf w_{j}}(\xb,t)\bigr]\).
Inequalities~\eqref{eq:score_bounds} imply
\(\|\mathbf d^{ij}(\xb,t)\|\le\bar d^{ij}\).
Applying Theorem~\ref{thm:wasserstein_contraction_bound} with nominal
\(\mathbf v_{\mathbf w_i}\), perturbed \(\mathbf v_{\mathbf w_j}\),
\(\mu_0=\nu_0\), \(\kappa=\kappa^{ij}\), \(\alpha=\alpha_{ij}\),
\(\bar d=\bar d^{ij}\) yields~\eqref{eq:detect_upper}.
If \(\mathbb W_2(\rho^{\mathbf w_{i}}(\cdot,T),\rho^{\mathbf w_{j}}(\cdot,T))\ge\epsilon\),
we obtain~\eqref{eq:detect_cert}, which follows from~\eqref{eq:detect_upper} at \(t=T\).
\end{proof}

The score-difference bound $\overline{s}_{ij}$ can be estimated empirically from trajectory data by fitting score networks to the fault-indexed marginals during offline training, making assumption~\eqref{eq:score_bounds} verifiable in practice.

\section{Fault-Indexed PF Operator Matching}
The contraction and detectability results of the previous section assume access to the exact fault-indexed PFO $\mathscr P_{s,t}^{\mathbf w}$. In practice, this operator is not available in closed-form. Hence, considering the offline trajectory data $\{(\xb_{s_i}^{(i)},\xb_{t_i}^{(i)},\mathbf w^{(i)})\}_{i=1}^N$ collected under each fault profile, our objective is to learn a parameterized flow map whose pushforward approximates $\mathscr P_{s,t}^{\mathbf w}$ for all $\mathbf w\in\mathscr F$. This has been achieved using flow map matching (FMM).

\label{sec:operator_learning}
For fixed \(\mathbf w\in\mathscr F\) and physical endpoints \(s<t\), let \(\rho_s^{\mathbf w}\) and \(\rho_t^{\mathbf w}\) denote the corresponding marginals. Following \cite{albergo2023building,albergo2023stochastic}, define the stochastic interpolant
\(I_\tau^{\mathbf w}
=
\alpha_\tau \xb_s+\beta_\tau \xb_t+\gamma_\tau \zb\), for \(\tau\in[0,1]\),
where \((\xb_s,\xb_t)\sim\pi_{s,t}^{\mathbf w} \in\Pi(\rho_s^{\mathbf w},\rho_t^{\mathbf w})\), \(\zb\sim\CMcal N(0,\Ib)\) is independent of \((\xb_s,\xb_t)\), and \(\alpha,\beta,\gamma\in C^1([0,1])\) satisfy the standard endpoint conditions. The pathwise velocity $\dot I_\tau^{\mathbf w}=\dot\alpha_\tau \xb_s+\dot\beta_\tau \xb_t+\dot\gamma_\tau \zb
$ is directly observable from data pairs \((\xb_s,\xb_t)\). Let \(\rho_\tau^{\mathbf w}\) denote the density of \(I_\tau^{\mathbf w}\).

Under standard regularity assumptions, the density curve \((\rho_\tau^{\mathbf w})_{\tau\in[0,1]}\) is reproduced by the auxiliary-time probability flow
\begin{equation}
\label{eq:fault_probability_flow}
\begin{aligned}
\frac{d}{d\tau}\xb(\tau)
&=
\mathbf b_\tau^{\mathbf w}(\xb(\tau)),
\\
\mathbf b_\tau^{\mathbf w}:\mathscr X\to\R^n,
&\
\mathbf b_\tau^{\mathbf w}(\xb)
:=
\mathbb E\!\left[\dot I_\tau^{\mathbf w}\mid I_\tau^{\mathbf w}=\xb\right].
\end{aligned}
\end{equation}
For each fixed \(0\le\tau_0\le\tau_1\le1\), let $\Phi_{\tau_0,\tau_1}^{\mathbf b^{\mathbf w}}:\mathscr X\to\mathscr X$ denote the corresponding two-time flow map. Its terminal map induces the exact fault-indexed PF operator
\(
\mathscr P_{s,t}^{\mathbf w}:\mathcal D(\mathscr X)\to\mathcal D(\mathscr X)\), where
\(\mathscr P_{s,t}^{\mathbf w}\rho
:=
\Bigl(\Phi_{0,1}^{\mathbf b^{\mathbf w}}\Bigr)_{\!\#}\rho.
\)
We will use $\mathscr P_{s,t}^{\mathbf w}$ for notational convenience whenever dependence on the probability flow is understood. Let us define the interpolant path map
\begin{equation*}
\label{eq:interpolant_path_map}
\Psi_{\tau_0,\tau_1}^{\mathbf w}:
\operatorname{supp}(\rho_{\tau_0}^{\mathbf w})\to\mathscr X,
\qquad
\Psi_{\tau_0,\tau_1}^{\mathbf w}\!\left(I_{\tau_0}^{\mathbf w}\right)
:=
I_{\tau_1}^{\mathbf w}
\qquad \text{a.s.}
\end{equation*}
The ODE flow \(\Phi^{\mathbf b^{\mathbf w}}\) propagates marginals and \(\Psi^{\mathbf w}\) is the samplewise transport observable from trajectory pairs. The FMM~\cite{boffi2024flow} supervises this pathwise object and uses the learned terminal map to induce a density operator.

The offline dataset consists of stochastic trajectory pairs $\{
(\xb_{s_i}^{(i)},\xb_{t_i}^{(i)},s_i,t_i,\mathbf w^{(i)})
\}_{i=1}^N$, for $\mathbf w^{(i)}\in\mathscr F$. We parameterize a learned fault-indexed two-time map
\(
\widehat\Phi_{\psi,\tau_0,\tau_1}^{\mathbf b^{\mathbf w}}:\mathscr X\to\mathscr X\),
for \(\widehat\Phi_{\psi,\tau,\tau}^{\mathbf b^{\mathbf w}}(\xb)=\xb,
\)
parameterized jointly by \((\tau_0,\tau_1,s,t,\mathbf w)\). Following the direct flow-map matching formulation of~\cite{boffi2024flow}, the direct FMM loss $\mathcal L_{\mathrm{FMM}}(\psi)$ can be defined as
\begin{equation*}
\label{eq:fault_fmm_loss}
\small
\begin{aligned}
\underset{\mathbf w, (s,t)}{\mathbb E} 
&\underset{[0,1]^2}{\int}
\omega(\tau_0,\tau_1)\,
\mathbb E\Big[
\big\|
\partial_{\tau_1}
\widehat\Phi_{\psi,\tau_0,\tau_1}^{\mathbf b^{\mathbf w}}
\big(
\widehat\Phi_{\psi,\tau_1,\tau_0}^{\mathbf b^{\mathbf w}}
(I_{\tau_1}^{\mathbf w})
\big)
-
\dot I_{\tau_1}^{\mathbf w}
\big\|^2
\\
&+
\big\|
\widehat\Phi_{\psi,\tau_0,\tau_1}^{\mathbf b^{\mathbf w}}
\big(
\widehat\Phi_{\psi,\tau_1,\tau_0}^{\mathbf b^{\mathbf w}}
(I_{\tau_1}^{\mathbf w})
\big)
-
I_{\tau_1}^{\mathbf w}
\big\|^2
\Big]
\,d\tau_0\,d\tau_1 ,
\end{aligned}
\end{equation*}
where \(\omega:[0,1]^2\to\R_{\ge0}\) has full support on \(\{(\tau_0,\tau_1):0\le~\tau_0\le\tau_1\le 1\}\). Supervision uses only observable interpolant velocities, and \(\mathbf b_\tau^{\mathbf w}\) is the theoretical generator. To enforce the semigroup structure, we add
\begin{equation*}
\label{eq:fault_semigroup_loss}
\small
\mathcal L_{\mathrm{sg}}(\psi)
:=
\mathbb E\!\left[
\left\|
\widehat\Phi_{\psi,\tau_0,\tau_2}^{\mathbf b^{\mathbf w}}(I_{\tau_0}^{\mathbf w})
-
\widehat\Phi_{\psi,\tau_1,\tau_2}^{\mathbf b^{\mathbf w}}
\!\left(
\widehat\Phi_{\psi,\tau_0,\tau_1}^{\mathbf b^{\mathbf w}}
(I_{\tau_0}^{\mathbf w})
\right)
\right\|^2
\right].
\end{equation*}
The terminal-path restriction of the FMM loss is also included in the following form:
\begin{equation}
\label{eq:fault_endpoint_loss}
\small
\mathcal L_{\mathrm{ep}}(\psi)
:=
\mathbb E_{\mathbf w,(s,t)}
\int_0^1
\omega_0(r)\,
\mathbb E\!\left[
\left\|
\partial_r\widehat\Phi_{\psi,0,r}^{\mathbf b^{\mathbf w}}(\xb_s)
-\dot I_r^{\mathbf w}
\right\|^2
\right]dr,
\end{equation}
where \((\xb_s,\xb_t)\sim\pi_{s,t}^{\mathbf w}\), \(I_0^{\mathbf w}=\xb_s\),
\(I_1^{\mathbf w}=\xb_t\), and \(\omega_0:[0,1]\to\R_{>0}\). One can then optimize
\(
\min_\psi\;
\mathcal L_{\mathrm{FMM}}(\psi)
+\lambda_{\mathrm{ep}}\mathcal L_{\mathrm{ep}}(\psi)
+\lambda_{\mathrm{sg}}\mathcal L_{\mathrm{sg}}(\psi).
\)
The semigroup term enforces two-time consistency. The endpoint term supervises the deployed terminal map \(\widehat\Phi_{\psi,0,1}^{\mathbf b^{\mathbf w}}\).

Under the hypotheses on \(\omega\) and on
\(\tau\mapsto\widehat\Phi_{\tau_0,\tau}^{\mathbf b^{\mathbf w}}
(I_{\tau_0}^{\mathbf w})\) stated in~\cite{boffi2024flow}, vanishing direct FMM loss \(\mathcal L_{\mathrm{FMM}}(\widehat\Phi)=0\) implies that the learned two-time maps agree with the interpolant path. In particular, \(\widehat\Phi_{\tau_0,\tau_1}^{\mathbf b^{\mathbf w}}
(I_{\tau_0}^{\mathbf w})=I_{\tau_1}^{\mathbf w}\) almost surely for
\(\omega\)-almost every \((\tau_0,\tau_1)\), and hence, \(\widehat\Phi_{0,1}^{\mathbf b^{\mathbf w}}(\xb_s)=\xb_t\) almost surely under \(\pi_{s,t}^{\mathbf w}\).

The learned terminal map induces $\widehat{\mathscr P}_{s,t}^{\mathbf w}:\mathcal D(\mathscr X)\to\mathcal D(\mathscr X)$, where $\widehat{\mathscr P}_{s,t}^{\mathbf w}\rho:=\left(\widehat\Phi_{\psi,0,1}^{\mathbf b^{\mathbf w}} \right)_{\!\#}\rho$. Because \(\widehat\Phi_{\psi,\tau_0,\tau_1}^{\mathbf b^{\mathbf w}}\) is parameterized jointly by \((\tau_0,\tau_1,s,t,\mathbf w)\), the learned PF-operator family defines a single conditional model in \(\mathbf w\), rather than a separate operator for each library element. Consequently, \(\widehat{\mathscr P}_{s,t}^{\mathbf w}\) can be evaluated for any continuous parameter set \(\mathbf w\in\mathcal W\), where \(\mathcal W\subseteq \mathbb{R}^p\) denotes a continuous fault-parameter set containing the discrete library \(\mathscr{F}\), beyond \(\mathscr F\).
\begin{lemma}[Induced PF Operator]
\label{lem:fault_flowmap_learns_pfo}
Fix \(\mathbf w\in\mathscr F\) and \(0\le s\le t\le T\). Set \(\widehat\Phi:=\widehat\Phi_{\psi,0,1}^{\mathbf b^{\mathbf w}}\). Assume \(\widehat\Phi\) is a \(C^1\)-diffeomorphism. Consequently, \(\widehat{\mathscr P}_{s,t}^{\mathbf w}\) is a valid density pushforward and
\(
\bigl(\widehat{\mathscr P}_{s,t}^{\mathbf w}\rho\bigr)(\xb)
=
\rho \bigl(
\widehat\Phi^{-1}(\xb)
\bigr)
\bigl|\det D\widehat\Phi^{-1}(\xb)\bigr|.
\)
If
\(\widehat\Phi=\Phi_{0,1}^{\mathbf b^{\mathbf w}}\), one has
\(
\widehat{\mathscr P}_{s,t}^{\mathbf w}
=
\mathscr P_{s,t}^{\mathbf w}\) and
\(\widehat{\mathscr P}_{s,t}^{\mathbf w}\rho_s^{\mathbf w}
=
\rho_t^{\mathbf w}.
\)
\end{lemma}

\begin{proof}
The first claim is the change-of-variables formula for the pushforward of densities by a \(C^1\)-diffeomorphism. If the learned terminal map equals the exact terminal map, the induced pushforwards coincide pointwise on every input density. Applying this identity to \(\rho_s^{\mathbf w}\) yields \(\widehat{\mathscr P}_{s,t}^{\mathbf w}\rho_s^{\mathbf w} = \rho_t^{\mathbf w}\).
\end{proof}

The next theorem converts the observable endpoint FMM residual into a \(\mathbb W_2\)-bound for the deployed operator.

\begin{theorem}[FMM Residual Controls PFO Error]
\label{thm:finite_error_operator_learning}
Fix \(\mathbf w\in\mathscr F\) and \(0\le s\le t\le T\). Write $\widehat\Phi_r:=\widehat\Phi_{\psi,0,r}^{\mathbf b^{\mathbf w}}$ with $
(X_s,X_t)\sim\pi_{s,t}^{\mathbf w}$ and let \(I_r^{\mathbf w}\) be the associated stochastic interpolant, so \(I_0^{\mathbf w}=X_s\) and \(I_1^{\mathbf w}=X_t\). Assume \(r\mapsto \widehat\Phi_r(X_s)\) is absolutely continuous and \(\widehat\Phi_0=\mathrm{Id}\). Define the endpoint residual
\begin{equation*}
\label{eq:endpoint_fmm_residual}
\begin{aligned}
R_r^{\mathbf w}
:=
\partial_r\widehat\Phi_r(X_s)-\dot I_r^{\mathbf w}, \
\mathcal E_{\mathrm{ep}}^{\mathbf w}(\psi)
:=
\int_0^1
\omega_0(r)\,
\mathbb E[\|R_r^{\mathbf w}\|^2]\,dr,
\end{aligned}
\end{equation*}
where \(\omega_0:[0,1]\to(0,\infty)\) is the terminal-path restriction sampling density from~\eqref{eq:fault_endpoint_loss}. Thus,
\begin{equation}
\label{eq:endpoint_residual_identity}
\widehat\Phi_{\psi,0,1}^{\mathbf b^{\mathbf w}}(X_s)-X_t
=
\int_0^1 R_r^{\mathbf w}\,dr
\qquad\text{a.s.}
\end{equation}
and therefore,
\begin{equation}
\label{eq:finite_error_operator_learning}
\mathbb W_2\!\left(
\widehat{\mathscr P}_{s,t}^{\mathbf w}\rho_s^{\mathbf w},\;
\mathscr P_{s,t}^{\mathbf w}\rho_s^{\mathbf w}
\right)
\le
\left(
\mathbb E\left[
\bigl\|
\widehat\Phi_{\psi,0,1}^{\mathbf b^{\mathbf w}}(X_s)-X_t 
\bigr\|^2
\right]
\right)^{1/2}.
\end{equation}
Moreover,
\[
\mathbb E\left[
\bigl\|
\widehat\Phi_{\psi,0,1}^{\mathbf b^{\mathbf w}}(X_s)-X_t
\bigr\|^2
\right]
\le
\int_0^1 \mathbb E[\|R_r^{\mathbf w}\|^2]\,dr.
\]
If \(\omega_0(r)\ge \underline\omega_0>0\) on \([0,1]\), one has
\begin{equation}
\label{eq:fmm_residual_operator_bound}
\mathbb W_2\!\left(
\widehat{\mathscr P}_{s,t}^{\mathbf w}\rho_s^{\mathbf w},\;
\mathscr P_{s,t}^{\mathbf w}\rho_s^{\mathbf w}
\right)
\le
\underline\omega_0^{-1/2}
\bigl(\mathcal E_{\mathrm{ep}}^{\mathbf w}(\psi)\bigr)^{1/2}.
\end{equation}
\end{theorem}

\begin{proof}
Since \(I_0^{\mathbf w}=X_s\), \(I_1^{\mathbf w}=X_t\), and \(\widehat\Phi_0(X_s)=X_s\), the fundamental theorem of calculus gives
\[
\widehat\Phi_1(X_s)-X_t
=
\int_0^1 \partial_r \widehat\Phi_r(X_s)\,dr
-\int_0^1 \dot I_r^{\mathbf w}\,dr
=
\int_0^1 R_r^{\mathbf w}\,dr,
\]
which is~\eqref{eq:endpoint_residual_identity}. Jensen's inequality then implies
\[
\mathbb E\left[\bigl\|
\widehat\Phi_1(X_s)-X_t
\bigr\|^2\right]
\le
\int_0^1 \mathbb E[\|R_r^{\mathbf w}\|^2]\,dr.
\]
Now set $\widehat X_t := \widehat\Phi_{\psi,0,1}^{\mathbf b^{\mathbf w}}(X_s)$. Then
\(\operatorname{Law}(\widehat X_t)=
\widehat{\mathscr P}_{s,t}^{\mathbf w}\rho_s^{\mathbf w}\)
and
\(\operatorname{Law}(X_t)=
\mathscr P_{s,t}^{\mathbf w}\rho_s^{\mathbf w}
=\rho_t^{\mathbf w}\),
so \(\operatorname{Law}(\widehat X_t,X_t)\) is a coupling of the learned and true terminal densities. By definition of \(\mathbb W_2\),
\[
\mathbb W_2^2\!\left(
\widehat{\mathscr P}_{s,t}^{\mathbf w}\rho_s^{\mathbf w},\;
\mathscr P_{s,t}^{\mathbf w}\rho_s^{\mathbf w}
\right)
\le
\mathbb E \left[\bigl\|
\widehat\Phi_{\psi,0,1}^{\mathbf b^{\mathbf w}}(X_s)-X_t
\bigr\|^2\right].
\]
Taking square roots yields~\eqref{eq:finite_error_operator_learning}. If \(\omega_0(r)\ge \underline\omega_0\), then
\begin{equation*}
\small
\begin{aligned}
\int_0^1 \mathbb E[\|R_r^{\mathbf w}\|^2]\,dr
\le
\underline\omega_0^{-1} \!
\int_0^1
\!\omega_0(r)\,
\mathbb E[\|R_r^{\mathbf w}\|^2]\,dr
=
\underline\omega_0^{-1}\mathcal E_{\mathrm{ep}}^{\mathbf w}(\psi),
\end{aligned}
\end{equation*}
which proves~\eqref{eq:fmm_residual_operator_bound}.
\end{proof}

Theorem~\ref{thm:finite_error_operator_learning} is the static training-marginal
counterpart to the long-horizon contraction bounds in Theorem~\ref{thm:wasserstein_contraction_bound} and Section~\ref{sec:contraction_learning}. It bounds one-shot operator mismatch from
\(\mathcal E_{\mathrm{ep}}^{\mathbf w}\).
Contraction regularization controls propagation under general input densities.

\section{Contractive Operator Learning}
\label{sec:contraction_learning}
Theorem~\ref{thm:finite_error_operator_learning} bounds operator error from the
FMM residual but not long-horizon stability.
We regularize the learned flow map
\(\widehat\Phi_{\psi,\tau_0,\tau_1}^{\mathbf b^{\mathbf w}}\) with a
contraction
certificate~\cite{Lohmiller1998,manchester2017,tsukamoto2021contraction}. Its
induced velocity field is defined by
\(
\partial_{\tau_1}
\widehat\Phi_{\psi,\tau_0,\tau_1}^{\mathbf b^{\mathbf w}}(\xb)
=
\widehat{\mathbf b}_{\psi,\tau_1}^{\mathbf w}
\!\left(
\widehat\Phi_{\psi,\tau_0,\tau_1}^{\mathbf b^{\mathbf w}}(\xb)
\right).
\)
Let \(\vartheta\) parameterize a metric network and set
\(
\mathbf M_{\vartheta}^{\mathbf w}(\xb,\tau)
=
\bigl(\Theta_{\vartheta}^{\mathbf w}(\xb,\tau)\bigr)^\top
\Theta_{\vartheta}^{\mathbf w}(\xb,\tau)\succ0,
\)
with optimized rate \(\alpha^{\mathbf w}\in\R\). Let us define the contraction
residual
\begin{equation}
\label{eq:fault_contraction_residual}
\mathcal R_{\mathrm{ctr}}^{\mathbf w}
:=
\dot{\mathbf M}_{\vartheta}^{\mathbf w}
\!+\bigl(D\widehat{\mathbf b}_{\psi,\tau}^{\mathbf w}\bigr)^\top
\mathbf M_{\vartheta}^{\mathbf w}
\!+\mathbf M_{\vartheta}^{\mathbf w}
D\widehat{\mathbf b}_{\psi,\tau}^{\mathbf w}
\!-2\alpha^{\mathbf w}\mathbf M_{\vartheta}^{\mathbf w},
\end{equation}
where \([\cdot]_+\) denotes projection onto the cone of symmetric positive
semidefinite matrices. The contraction and certificate losses are
\begin{align*}
\mathcal L_{\mathrm{ctr}}(\psi,\vartheta,\alpha)
&:=
\mathbb E\!\left[
\left\|
\left[
\mathcal R_{\mathrm{ctr}}^{\mathbf w}(\xb,\tau)
\right]_+
\right\|_F^2
\right],\\
\mathcal L_{\mathrm{cert}}(\vartheta,\alpha)
&:=
\mathbb E_{\mathbf w}\!\left[
c_\alpha[\alpha^{\mathbf w}]_+^2
+c_\kappa
\bigl(
\log\overline m^{\mathbf w}
-\log\underline m^{\mathbf w}
\bigr)
\right],
\end{align*}
with uniform metric bounds
\(
\underline m^{\mathbf w}\Ib
\preceq
\mathbf M_{\vartheta}^{\mathbf w}
\preceq
\overline m^{\mathbf w}\Ib
\).
The joint objective is
\begin{equation}
\label{eq:full_obj}
\begin{aligned}
\min_{\psi,\vartheta,\alpha}\;\;
&\mathcal L_{\mathrm{FMM}}(\psi)
+\lambda_{\mathrm{ep}}\mathcal L_{\mathrm{ep}}(\psi)
+\lambda_{\mathrm{sg}}\mathcal L_{\mathrm{sg}}(\psi) \\
&+\lambda_{\mathrm{ctr}}\mathcal L_{\mathrm{ctr}}(\psi,\vartheta,\alpha)
+\lambda_{\mathrm{cert}}\mathcal L_{\mathrm{cert}}(\vartheta,\alpha).
\end{aligned}
\end{equation}
The weights \(\lambda_{\mathrm{ep}},\lambda_{\mathrm{sg}},\lambda_{\mathrm{ctr}},\lambda_{\mathrm{cert}}>0\)
balance five terms. The
\(\mathcal L_{\mathrm{ep}}\) supervises the terminal map and feeds
Theorem~\ref{thm:finite_error_operator_learning} and
\(\mathcal L_{\mathrm{ctr}}\) drives the contraction residual to zero. Moreover,
\(\mathcal L_{\mathrm{cert}}\) penalizes \(e^{\alpha^{\mathbf w}}\) and
\(\kappa^{\mathbf w}=\sqrt{\overline m^{\mathbf w}/\underline m^{\mathbf w}}\)
and shrinks the bias \(\delta^{\mathbf w}\) in~\eqref{eq:delta_w}.

One can write the exact auxiliary field as a perturbation of the learned one, such that
\(\mathbf b_{\tau}^{\mathbf w}(\xb)
=
\widehat{\mathbf b}_{\psi,\tau}^{\mathbf w}(\xb)
+\mathbf d_{\tau}^{\mathbf w}(\xb).
\)

\begin{theorem}[Approximate Wasserstein bound]
\label{thm:approx_ctr}
Fix \(\mathbf w\in\mathscr F\). Set
\(\widehat\Phi_\tau:=\widehat\Phi_{\psi,0,\tau}^{\mathbf b^{\mathbf w}}\), and
assume \(\tau\mapsto\widehat\Phi_\tau(\xb)\) is \(C^1\) for each \(\xb\),
\(\mathcal R_{\mathrm{ctr}}^{\mathbf w}\) is continuous,
\(\rho_\tau^{\mathbf w}\) has full support on \(\mathscr X\), and
\[
\underline m^{\mathbf w}\Ib
\preceq
\mathbf M_{\vartheta}^{\mathbf w}(\xb,\tau)
\preceq
\overline m^{\mathbf w}\Ib
\qquad
\forall (\xb,\tau)\in\mathscr X\times[0,1].
\]
Let
\begin{equation}
\label{eq:eps_ctr}
\varepsilon_{\mathrm{ctr}}^{\mathbf w}
:=
\sup_{(\xb,\tau)\in\mathscr X\times[0,1]}
\bigl\|\bigl[\mathcal R_{\mathrm{ctr}}^{\mathbf w}(\xb,\tau)\bigr]_+\bigr\|_F
\ge0,
\end{equation}
and define
\begin{equation}
\label{eq:alpha_tilde}
\tilde\alpha^{\mathbf w}
:=
\alpha^{\mathbf w}
+\frac{\varepsilon_{\mathrm{ctr}}^{\mathbf w}}{2\,\underline m^{\mathbf w}}.
\end{equation}
Suppose \(\mathbf b_\tau^{\mathbf w}(\xb)=\widehat{\mathbf b}_{\psi,\tau}^{\mathbf w}(\xb)+\mathbf d_\tau^{\mathbf w}(\xb)\) with
\(\overline{d}^{\mathbf w}:=\sup_{\xb,\tau}\|\mathbf d_\tau^{\mathbf w}(\xb)\|<\infty\).
Then, along the auxiliary flow,
\[
\dot{\mathbf M}_{\vartheta}^{\mathbf w}
+2\operatorname{sym}\!\left(
\mathbf M_{\vartheta}^{\mathbf w}
D\widehat{\mathbf b}_{\psi,\tau}^{\mathbf w}
\right)
\preceq
2\tilde\alpha^{\mathbf w}\mathbf M_{\vartheta}^{\mathbf w},
\]
and with \(\kappa^{\mathbf w}:=\sqrt{\overline m^{\mathbf w}/\underline m^{\mathbf w}}\) and
\(
\tilde\delta_\tau^{\mathbf w}
:=
\kappa^{\mathbf w}\overline{d}^{\mathbf w}
\int_0^\tau e^{\tilde\alpha^{\mathbf w}(\tau-r)}\,dr
\),
Theorem~\ref{thm:wasserstein_contraction_bound} applies with rate
\(\tilde\alpha^{\mathbf w}\) and perturbation level \(\overline{d}^{\mathbf w}\).
Thus, for any \(\mu_0,\nu_0\in\mathcal P_2(\mathscr X)\) and \(\tau\in[0,1]\),
\begin{equation}
\label{eq:flow_map_bound}
\begin{aligned}
\mathbb W_2 \left(
(\widehat\Phi_{\psi,0,\tau}^{\mathbf b^{\mathbf w}})_{\!\#}\mu_0,\;
(\Phi_{0,\tau}^{\mathbf b^{\mathbf w}})_{\!\#}\nu_0
\right)
\le{}&
\kappa^{\mathbf w}e^{\tilde\alpha^{\mathbf w}\tau}
\mathbb W_2(\mu_0,\nu_0)
+\tilde\delta_\tau^{\mathbf w}.
\end{aligned}
\end{equation}
In particular, for any \(\rho_s,\eta_s\in\mathcal D(\mathscr X)\), one has
\begin{equation}
\label{eq:operator_bound}
\mathbb W_2\!\left(
\widehat{\mathscr P}_{s,t}^{\mathbf w}\rho_s,\;
\mathscr P_{s,t}^{\mathbf w}\eta_s
\right)
\le
\kappa^{\mathbf w}e^{\tilde\alpha^{\mathbf w}}
\mathbb W_2(\rho_s,\eta_s)
+\tilde\delta_1^{\mathbf w}. 
\end{equation}
For the same input density \(\rho_s\), define the terminal bias
\begin{equation}
\label{eq:delta_w}
\delta^{\mathbf w}
:=
\tilde\delta_1^{\mathbf w}
=
\kappa^{\mathbf w}d^{\mathbf w}
\int_0^1 e^{\tilde\alpha^{\mathbf w}(1-r)}\,dr.
\end{equation}
Consequently, we have
\begin{equation}
\label{eq:operator_same_input_bound}
\mathbb W_2\!\left(
\widehat{\mathscr P}_{s,t}^{\mathbf w}\rho_s,\;
\mathscr P_{s,t}^{\mathbf w}\rho_s
\right)
\le
\delta^{\mathbf w}.
\end{equation}
\end{theorem}

\begin{proof}
The definition~\eqref{eq:fault_contraction_residual} and the Frobenius bound~\eqref{eq:eps_ctr} on \([\mathcal R_{\mathrm{ctr}}^{\mathbf w}]_+\) imply the Jacobian analogue of Lemma~\ref{lem:deterministic_contraction} along the learned auxiliary flow with rate
\(\tilde\alpha^{\mathbf w}\) in~\eqref{eq:alpha_tilde}, as in the standard quadratic-form argument~\cite{Lohmiller1998,tsukamoto2021contraction}.
Applying Theorem~\ref{thm:wasserstein_contraction_bound} on \(\tau\in[0,1]\) with
nominal field \(\widehat{\mathbf b}_{\psi,\tau}^{\mathbf w}\), perturbation
\(\mathbf d_\tau^{\mathbf w}\), homotopy~\eqref{eq:particular_solutions_virtual_system},
rate \(\tilde\alpha^{\mathbf w}\), and \(\kappa^{\mathbf w}\), \(\overline{d}^{\mathbf w}\)
as above yields~\eqref{eq:flow_map_bound}.
The operator bound~\eqref{eq:operator_bound} follows at \(\tau=1\), and~\eqref{eq:delta_w}--\eqref{eq:operator_same_input_bound}
specialize~\eqref{eq:operator_bound} to \(\eta_s=\rho_s\).
\end{proof}

Theorem~\ref{thm:approx_ctr} gives the main long-horizon \(\mathbb W_2\) error bound for the learned operator family under contraction regularization.
\begin{corollary}[Exact contraction regularization]
\label{cor:exact_ctr}
If \(\mathcal L_{\mathrm{ctr}}(\psi,\vartheta,\alpha)=0\) under the hypotheses of
Theorem~\ref{thm:approx_ctr}, then \(\varepsilon_{\mathrm{ctr}}^{\mathbf w}=0\),
\(\tilde\alpha^{\mathbf w}=\alpha^{\mathbf w}\), and
\(\tilde\delta_\tau^{\mathbf w}\) coincides with the bias term obtained by
replacing \(\tilde\alpha^{\mathbf w}\) with \(\alpha^{\mathbf w}\) throughout
Theorem~\ref{thm:approx_ctr}.
\end{corollary}

\section{Fault Inference and Recovery Control}
We extend the fault-indexed density-transport framework of Sections~\ref{sec:operator_learning} and~\ref{sec:contraction_learning} to fault identification and recovery via density tracking control. The learned operator library $\mathfrak P_{\mathrm{lib}} := \{\widehat{\mathscr P}_{s,t}^{\mathbf w}:\mathbf w\in\mathscr F\}$ is computed offline and queried online for both inference and recovery.

\subsection{Fault Inference}
Recall that \(\mathscr F=\{\mathbf w_{0},\ldots,\mathbf w_{N_f}\}\) is the fault library and let
\(
\widehat\rho_{k|k}^{(j)} = \frac{1}{N_p}\sum_{i=1}^{N_p}\delta_{\xb_{k|k}^{i,(j)}}
\)
denote the posterior under hypothesis \(\mathbf w_{j}\) for $j \in\left\{0, \ldots, N_f\right\}$. For each  $\mathbf{w}_j$, prediction uses the corresponding fault-indexed flow map,
\begin{equation*}
\label{eq:online_bank_prediction}
\xb_{k+1|k}^{i,(j)}
=
\Phi_{t_k,t_{k+1}}^{\mathbf F_{\mathbf w_{j}}}
\!\bigl(\xb_{k|k}^{i,(j)}\bigr).
\end{equation*}
At time \(t_k\), let \(\CMcal Y_k:=\{\yb_0,\ldots,\yb_k\}\) such that for each candidate \(\mathbf w\), the one-step predicted density and predictive likelihood are with observations \(\yb_k=h(\xb_k)+\boldsymbol\epsilon_k\), where \(h:\mathscr X\to\R^{n_y}\) is the observation map and \(\boldsymbol\epsilon_k\sim p_\epsilon\) is i.i.d.\ measurement noise with known density \(p_\epsilon\) on \(\R^{n_y}\). The conditional measurement density is \(p(\yb\mid\xb)=p_{\boldsymbol\epsilon}(\yb-h(\xb))\). Denote the current filtered density by \(\widehat\rho_{k|k}^{(j)}\in\mathcal D(\mathscr X)\). The predictive likelihood is then
\begin{align}
\label{eq:fault_density_prediction}
\widehat{\rho}_{k+1 \mid k}^{(j)} &:=\widehat{\mathscr{P}}_{t_k, t_{k+1}}^{\mathbf{w}_j} \widehat{\rho}_{k \mid k}^{(j)},
\\
\label{eq:fault_predictive_likelihood}
\mathcal L^{(j)}_{k+1}
&:=
\int_{\mathscr X}
p(\yb_{k+1}\mid\xb)\,
\widehat\rho_{k+1|k}^{(j)}(\xb)\,d\xb.
\end{align}
If \(p_\epsilon = \mathcal N (0, \mathbf R) \), i.e., \(\boldsymbol \epsilon_k \sim \mathcal N (0, \mathbf R)\), one has
\begin{equation}
\mathcal L_{k+1}^{(j)} \propto \frac1{N_p}\sum_{i=1}^{N_p} \exp\!\left(-\frac12 \|h(\xb_{k+1|k}^{i,(j)})-\yb_{k+1}\|_{\mathbf R^{-1}}^2 \right).
\end{equation}
Instead of selecting a discrete library, we perform iterative fault inference over a continuous parameter set \(\mathcal W\) using the maximum-likelihood estimator (MLE)
\begin{equation}
\label{eq:online_continuous_fault_mle}
\widehat{\mathbf w}_{\mathrm{MLE},k}
\in
\argmin_{\mathbf w \in \mathcal{W}}
\sum_{\ell=0}^{k}
\frac12\|h(\xb_\ell(\mathbf w))-\yb_\ell\|_{\mathbf R^{-1}}^2
\end{equation}
subject to the nonlinear dynamics
\begin{equation*}
\label{eq:online_continuous_fault_constraint}
\xb_{\ell+1}(\mathbf w)
=
\Phi_{t_\ell,t_{\ell+1}}^{\mathbf F_{\mathbf w}}
\!\left(\xb_\ell(\mathbf w)\right),
\qquad
\xb_0(\mathbf w)=\xb_0.
\end{equation*}
Since \eqref{eq:online_continuous_fault_mle} optimizes over \(\mathbf w\in\mathcal W\), the estimate \(\widehat{\mathbf w}_{\mathrm{MLE},k}\) is not restricted to \(\mathscr F\). Thus, online fault inference can be viewed as a continuous fault-fitting problem over \(\mathcal W\), with density reachability estimated by evaluating the learned fault-indexed PF-operator family at \(\widehat{\mathbf w}_{\mathrm{MLE},k}\).

With positive prior mass on the true fault and predictive likelihoods separating the hypotheses (Definition~\ref{def:density_reachability}), the posterior is expected to concentrate on the best matching operator, consistent with standard Bayesian concentration~\cite{ghosal2017fundamentals}. The predicted family \(\{\widehat\rho_{k+1|k}^{(j)}\}_{j=0}^{N_f}\) is the
learned analog of~\eqref{eq:reachable_density_family}, so fault detection, identifiability, and recovery are evaluated in the density space through density reachability.

\subsection{Recovery via Contractive GMM control}
Let $\hat{\rho}_{k \mid k}:=\hat{\rho}_{k \mid k}^{\widehat{\mathrm{w}}_{\mathrm{MLE}, k}}$ where $\widehat{\mathrm{w}}_{\mathrm{MLE}, k} \in \mathcal{W}$. Once the posterior concentrates, we propagate the current filtered density over a finite recovery horizon \(\{t_k,\ldots,t_{k+N_r}\}\), where \(N_r\in\mathbb N\), and \(\mathbb{N}\) is the set of positive integers, under both the inferred fault operator and the nominal operator:
\begin{align}
\label{eq:fault_recovery_density}
\widehat\rho_{k+\ell|k}^{\mathrm{fault}}
&:=
\widehat{\mathscr P}_{t_k,t_{k+\ell}}^{\widehat{\mathbf w}_{\mathrm{MLE},k}}
\widehat\rho_{k|k},
\\
\label{eq:nominal_recovery_density}
\widehat\rho_{k+\ell|k}^{\mathrm{nom}}
&:=
\widehat{\mathscr P}_{t_k,t_{k+\ell}}^{\mathbf w_0}
\widehat\rho_{k|k},
\qquad
\ell=0,\ldots,N_r.
\end{align}
Here~\eqref{eq:fault_recovery_density} is the estimated fault reachable density and~\eqref{eq:nominal_recovery_density} the nominal recovery target.

\begin{lemma}[Recovery target surrogate bound]
\label{lem:recovery_surrogate}
Assume the true active fault on \([t_k,t_{k+N_r}]\) is \(\mathbf w_\star=\widehat{\mathbf w}_{\mathrm{MLE}, k}\), let \(\rho_{k|k}\) denote the true filtered density, and define the true terminal fault reachable density
\(
\rho_{k+N_r}^{\mathrm{fault}}
:=
\mathscr P_{t_k,t_{k+N_r}}^{\mathbf w_\star}\rho_{k|k}.
\)
Suppose the Gaussian Mixture Model (GMM) approximation of the estimated terminal fault reachable density satisfies, for some \(\varepsilon_f>0\), we have
\(
\mathbb W_2\!\left(
\widehat\rho_{k+N_r|k}^{\mathrm{fault}},
\rho_{\mathrm{GMM}}^{\mathrm{fault}}
\right)\le\varepsilon_f\).
Consequently, one has
\begin{equation*}
\label{eq:recovery_surrogate_bound}
\mathbb W_2\!\left(
\rho_{k+N_r}^{\mathrm{fault}},
\rho_{\mathrm{GMM}}^{\mathrm{fault}}
\right)
\le
\kappa^{\mathbf w_\star}e^{\alpha^{\mathbf w_\star}}
\mathbb W_2(\rho_{k|k},\widehat\rho_{k|k})
+\delta_1^{\mathbf w_\star}
+\varepsilon_f.
\end{equation*}
\end{lemma}

\begin{proof}
By the triangle inequality, one has
\[
\mathbb W_2\!\left(
\rho_{k+N_r}^{\mathrm{fault}},
\rho_{\mathrm{GMM}}^{\mathrm{fault}}
\right)
\le
\mathbb W_2\!\left(
\rho_{k+N_r}^{\mathrm{fault}},
\widehat\rho_{k+N_r|k}^{\mathrm{fault}}
\right)
+\varepsilon_f.
\]
Applying~\eqref{eq:operator_bound} with \(\mathbf w=\mathbf w_\star\), \(\rho_s=\rho_{k|k}\), and \(\eta_s=\widehat\rho_{k|k}\) gives the stated bound.
\end{proof}

To synthesize recovery control, let us approximate the terminal fault and nominal reachable densities by matched-weight Gaussian mixtures with \(M\in\mathbb N\) components, $\widehat\rho_{k+N_r|k}^{\mathrm{fault}}(\xb) \approx \sum_{i=1}^M\beta_i\,\CMcal N(\xb;m_i^f,\Sigma_i^f)$ and $ \widehat\rho_{k+N_r|k}^{\mathrm{nom}}(\xb) \approx \sum_{i=1}^M\beta_i\,\CMcal N(\xb;m_i^n,\Sigma_i^n)$. Matched weights preserve component correspondence, so recovery is posed componentwise from \((m_i^f,\Sigma_i^f)\) to \((m_i^n,\Sigma_i^n)\). Let
\begin{equation*}
\label{eq:gmm_responsibility_recovery}
\gamma_i(\xb_k)
:=
\frac{
\beta_i\,\varphi(\xb_k;m_i^f,\Sigma_i^f)
}{
\sum_{r=1}^{M}\beta_r\,\varphi(\xb_k;m_r^f,\Sigma_r^f)
}
\end{equation*}
denote the fault-component responsibility at the current state \(\xb_k\). For each retained component \(i\in\mathcal I_k\), linearizing about the faulted component mean yields the local model
\begin{equation*}
\label{eq:gmm_ccm_local_model}
\xb_{\ell+1}
=
A_\ell^{(i)}\xb_\ell+B_\ell^{(i)}\delta\ub_\ell+\mathbf{c}_\ell^{(i)},
\qquad
\ell=0,\dots,N_r-1,
\end{equation*}
with affine law $\delta\ub_\ell^{(i)} = \boldsymbol{\nu}_\ell^{(i)} + K_\ell^{(i)}\bigl(\xb_\ell-\mu_\ell^{(i)}\bigr)$. The induced mean--covariance recursion is
\begin{equation*}
\label{eq:gmm_ccm_moment_recursion}
\begin{aligned}
\mu_{\ell+1}^{(i)}
&=
A_\ell^{(i)}\mu_\ell^{(i)}
+
B_\ell^{(i)}\boldsymbol{\nu}_\ell^{(i)}
+
\mathbf{c}_\ell^{(i)},
\\
\Sigma_{\ell+1}^{(i)}
&=
\bigl(A_\ell^{(i)}+B_\ell^{(i)}K_\ell^{(i)}\bigr)\Sigma_\ell^{(i)}
\bigl(A_\ell^{(i)}+B_\ell^{(i)}K_\ell^{(i)}\bigr)^\top
+
W_\ell^{(i)},
\end{aligned}
\end{equation*}
initialized at $\mu_0^{(i)}=m_i^f$ and $\Sigma_0^{(i)}=\Sigma_i^f$ where \(W_\ell^{(i)}\succeq0\) is the residual covariance of the propagated fault tube. 

Define $Q_{\ell} \succeq 0$, $R_{\ell}\succ 0$, and $P_{N_r}=\lambda_T Q_{N_r-1}$. We can compute a contraction metric \(P_\ell\) backward along the nominal linearization by
\begin{equation}
\label{eq:gmm_ccm_riccati}
\begin{split}
P_\ell
=
Q_\ell &
+
(A_\ell^{\mathrm{nom}})^\top P_{\ell+1}A_\ell^{\mathrm{nom}}
- \\
&(A_\ell^{\mathrm{nom}})^\top P_{\ell+1}B_\ell^{\mathrm{nom}}
S_\ell^{-1}
(B_\ell^{\mathrm{nom}})^\top P_{\ell+1}A_\ell^{\mathrm{nom}},
\end{split}
\end{equation}
where $S_\ell:=R_\ell+(B_\ell^{\mathrm{nom}})^\top P_{\ell+1}B_\ell^{\mathrm{nom}}$. It can be shown by application of the converse contraction theorem applied to control contraction metrics~\cite{manchester2017,giesl2022review} that~\eqref{eq:gmm_ccm_riccati} defines a valid contraction metric. 

For each \(i\in\mathcal I_k\), we solve
\begin{equation*}
\label{eq:gmm_ccm_component_ocp}
\min_{\{\boldsymbol{\nu}_\ell^{(i)},K_\ell^{(i)}\}_{\ell=0}^{N_r}}
J^{(i)}
:=
\sum_{\ell=0}^{N_r}
\Bigl[
\omega_\ell \,\ell_{\mathrm{trk},\ell}^{(i)}
+
\rho_\nu\|\boldsymbol{\nu}_\ell^{(i)}\|_2^2
+
\rho_K\|K_\ell^{(i)}\|_F^2
\Bigr],
\end{equation*}
where
\begin{equation*}
\begin{aligned}
\ell_{\mathrm{trk},\ell}^{(i)}:= q_m&\|\mu_{\ell+1}^{(i)}\!-m_i^n\|_2^2
+
q_\Sigma\|\Sigma_{\ell+1}^{(i)}\!-\Sigma_i^n\|_F^2
\\ 
&+\lambda_m\|\mu_{\ell+1}^{(i)}\!-\xb_{k+\ell+1}^{\mathrm{nom}}\|_{P_{\ell+1}}^2
+
\lambda_\Sigma\operatorname{Tr}\!\bigl(P_{\ell+1}\Sigma_{\ell+1}^{(i)}\bigr).
\end{aligned}
\end{equation*}
Here $\omega_\ell \geq 0$ is a stage weight, and $q_m, q_{\Sigma}, \lambda_m, \lambda_{\Sigma}, \rho_\nu, \rho_K > 0$ are the weights for mean matching, covariance matching, contraction-metric regularization, and control effort. The controls are then the blended normalized responsibilities,
\begin{equation*}
\label{eq:gmm_ccm_first_step}
\delta\ub_k^\star
=
\sum_{i\in\mathcal I_k}\bar\gamma_i(\xb_k)
\Bigl(
\boldsymbol{\nu}_0^{(i)}
+
K_0^{(i)}(\xb_k-m_i^f)
\Bigr),
\end{equation*}
with $\bar\gamma_i(\xb_k)
:=
\frac{\gamma_i(\xb_k)}{\sum_{r\in\mathcal I_k}\gamma_r(\xb_k)}$. Lemma~\ref{lem:recovery_surrogate} bounds the gap between the true and GMM-approximated fault reachable densities, justifying the use of the learned target. The recovery correction is applied iteratively to the baseline controller as $\ub_k=\ub_{\mathrm{cl},k}+\ub_{\mathrm{rec},k}$ where $\ub_{\mathrm{rec},k}:= \delta\ub_k^\star$.

\section{Simulations}
We validate the proposed contractive Perron--Frobenius reachability FDI and recovery framework on a 10-state spacecraft attitude-control benchmark with four tetrahedrally mounted reaction wheels, where $x=\left[\theta^{\top}, \omega^{\top}, \omega_w^{\top}\right]^{\top}$ collects attitude, body rates, and wheel speeds. The spacecraft parameters follow~\cite{lee2017geometric} with $I=\operatorname{diag}(1.0,1.0,0.8)\,\mathrm{kg\,m}^2$ and $J_w=0.01\,\mathrm{kg\,m}^2$. The nominal controller is a saturated PD tracker with $K_p=\operatorname{diag}(22.5,18.0,15.0)$, $K_d=\operatorname{diag}(12.0,9.0,7.5)$, wheel-torque limit $0.14$ Nm, and sampling time $\Delta t=0.02\,\mathrm{s}$, designed to follow $\theta_d(t)= [0.05 \sin (0.2 \pi t), 0.05 \cos (0.2 \pi t),(\pi / 250) t]^{\top}$ via $u_{\text {nom }}=-K_p\left(\theta-\theta_d\right)-K_d \omega$ and $u_w=\operatorname{sat}\left(A^{\dagger} u_{\text {nom }}, 0.14\right)$. Faults are modeled as constant wheel loss-of-effectiveness coefficients $\alpha \in[0,1]^4$. The filter uses full-state noisy measurements, and recovery combines the learned fault-indexed PF operator, continuous OOD fault refinement, and contractive GMM replanning under matched nonlinear rollout-noise realizations.

For the OOD case $\alpha_{\star}=[0.15,0.4,0.2,0.25]^{\top}$ over a $10\,\mathrm{s}$ horizon (500 steps), the inferred fault is $\widehat{\alpha}=[0.14972,0.40174,0.19999,0.24871]^{\top}$, with $\ell_2$ error $2.19 \times 10^{-3}$. Over 25 uniformly sampled OOD faults in $[0,1]^4$, the continuous estimator attains mean fault error $7.34\times 10^{-3}$ in $\ell_2$, which indicates accurate generalization beyond the discrete training library. The empirical terminal reachable-density separation is \(4.38\times10^{-1}\) in \(\mathbb W_2\), which is consistent with the sampled rollout-restricted median bound \(6.48\) from~\eqref{eq:detect_cert} over one inference step \(\Delta t=0.02\,\mathrm{s}\). The faulted, i.e., without recovery control, and recovered trajectories are driven by the same state-dependent multiplicative noise realization on the body rates and wheel speeds, whereas the nominal trajectory remains noise-free. Using the proposed FDIR, the terminal full-state error relative to nominal drops from $1.058 \times 10^{1}$ under fault to $9.909 \times 10^{-1}$ under recovery, a $10.68 \times$ improvement, while the mean full-state error decreases from $7.478$ to $5.745 \times 10^{-1}$, a $13.02 \times$ pathwise improvement. Figure~\ref{fig:ood_fdir_performance}(b) further shows that the continuous fault-estimation error decreases to $8.18\times 10^{-2}$, and the empirical terminal operator gap \(\mathbb W_2(\widehat{\mathscr P}_{0,T}^{\mathbf w}\rho_0,\mathscr P_{0,T}^{\mathbf w}\rho_0)\approx 3.02\) remains below the sampled rollout-restricted worst-case bound \(\delta^{\mathbf w}\approx 2.26\times10^{2}\) from~\eqref{eq:delta_w}.

\begin{figure}[t]
\centering
\includegraphics[width=\columnwidth,trim=0 0 0 0,clip]{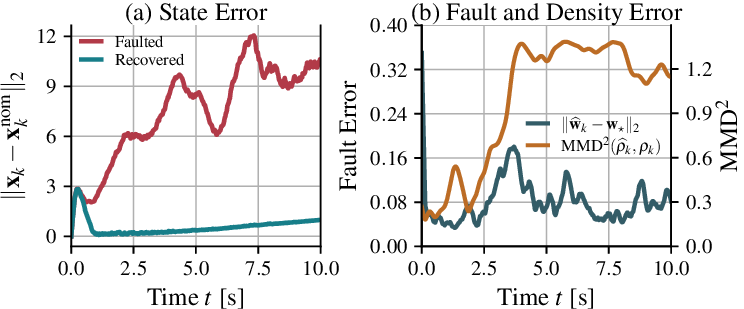}
\caption{OOD FDIR performance for $\alpha_{\star}=[0.15,0.4,0.2,0.25]^{\top}$. (a) The reduction in full-state tracking error under recovery control. (b) Continuous fault-estimation error and reachable-density discrepancy $\mathrm{MMD}^2(\widehat{\rho}_k,\rho_k)$.}   
\label{fig:ood_fdir_performance}
\end{figure} 
\section{Conclusion}
This paper presented a Perron-Frobenius (PF) operator-based framework for data-driven fault detection, identification, and recovery (FDIR) in nonlinear stochastic systems via density reachability. Using probability-flow PF operators and contraction theory, we defined fault-indexed reachable density families and established $\mathbb{W}_2$ conditions for density-based fault detectability and identifiability. We utilized flow map matching to learn the fault-indexed operators from trajectory data, and demonstrated that the observable endpoint residual directly bounds the deployed operator error in $\mathbb{W}_2$. Furthermore, we co-trained a contraction certificate that provides explicit long-horizon guarantees relative to the true fault-driven and nominal density evolutions. The learned operator library was then used online for recursive Bayesian fault inference and continuous fault parameter fitting to generalize the learned map to out-of-distribution (OOD) scenarios. To carry out the recovery control, we employed reachable-density propagation and Gaussian-mixture covariance steering. The proposed density reachability-based framework in this paper provides a systematic framework in which stochastic diagnosis, operator learning, and recovery control can be posed and analyzed together.

\bibliographystyle{IEEEtran}
\bibliography{sample} 
\end{document}